\documentclass[runningheads]{scrartcl} %arxiv
\usepackage[T1]{fontenc}
\usepackage{amsmath}
\usepackage{amssymb}
\usepackage{graphicx}
\usepackage{cite}
\usepackage{fullpage}
\usepackage[usenames]{color}
\usepackage[colorlinks=true]{hyperref}
\usepackage{float}
\usepackage{xspace}
\usepackage{mathtools}
\usepackage{comment}

\newcommand{\fb}[2]{f\ensuremath{(#1, #2)}\xspace}
\newcommand{\fa}[1]{f\ensuremath{(#1)}\xspace}
\newcommand{\gb}[2]{g\ensuremath{(#1, #2)}\xspace}
\newcommand{\ga}[1]{g\ensuremath{(#1)}\xspace}
\newcommand{\fFunc}{\ensuremath{f(h)}\xspace}
\newcommand{\gFunc}{\ensuremath{g(h)}\xspace}
\newcommand{\fFuncI}[1]{\ensuremath{f_{#1}(h)}\xspace}
\newcommand{\gFuncI}[1]{\ensuremath{g_{#1}(h)}\xspace}
\newcommand{\unseen}[1]{\ensuremath{T_{#1}(h)}}
\newcommand{\subT}[2]{\ensuremath{T(#1,#2)}} %subchain
\newcommand{\altitude}[1]{\ensuremath{L(#1)}}
\newcommand{\ymax}{\ensuremath{y(T)}}
\newcommand{\fullInterval}{\ensuremath{[\ymax, \infty)}}
\newcommand{\interval}[1]{\ensuremath{I_{#1}}}
\newcommand{\initInterval}{\ensuremath{\interval{0}}}
\newcommand{\g}[1]{\ensuremath{u_{#1}}}
\newcommand{\full}[1]{\ensuremath{U_{#1}(h)}}
\newcommand{\fullLeft}[1]{\ensuremath{U_{#1}(h)}}
\newcommand{\fullRight}[1]{\ensuremath{\overline{U}_{#1}(h)}}
\newcommand{\partialLeft}[1]{\ensuremath{V_{#1}(h)}}
\newcommand{\partialRight}[1]{\ensuremath{\overline{V}_{#1}(h)}}
\newcommand{\peak}[2]{\ensuremath{\textsf{peak}(#1,#2)}}
\newcommand{\SPT}[1]{\ensuremath{\textsf{SPT}(#1)}\xspace} %Shortest path tree
\newcommand{\guarded}[1]{\textsf{Vis}\ensuremath{(#1)}\xspace}
\newcommand{\SRF}[2]{\ensuremath{\Pi_{#1,#2}}} % set of rational function
\newcommand{\IHP}[2]{\ensuremath{\mathcal{E}(#1,#2)}} % intersection of half-planes
\newcommand{\IPHP}[2]{\ensuremath{\mathcal{E}^+(#1,#2)}} % intersection of positive slope half-planes
\newcommand{\IPHPInterval}[1]{\ensuremath{P(#1)}}
\newcommand{\INHP}[2]{\ensuremath{\mathcal{E}^-(#1,#2)}} % intersection of negative slope half-planes

\newtheorem{lemma}{Lemma}
\newtheorem{theorem}{Theorem}

\newbox\ProofSym
\setbox\ProofSym=\hbox{%
   \unitlength=0.18ex%
   \begin{picture}(10,10)
   \put(0,0){\framebox(9,9){}}
   \put(0,3){\framebox(6,6){}}
   \end{picture}}

\makeatother

\let\geq\geqslant
\let\leq\leqslant
\let\ge\geqslant
\let\le\leqslant

\newenvironment{denseitems}{\list{$\bullet$}%
  {\labelwidth3em\itemsep0pt\parsep0pt\topsep0.6ex}}{\endlist}
  
\begin{document}
\title{Guarding Terrains with Guards on a Line}

\author{Byeonguk Kang\thanks{Department of Computer Science and Engineering, Pohang University of Science and Technology, Pohang, Korea. \texttt{\{kbu417,hwikim\}@postech.ac.kr}}
\and Hwi Kim\footnotemark[1]
\and Hee-Kap Ahn\thanks{Department of Computer Science and Engineering, Graduate School of Artificial Intelligence, 
Pohang University of Science and Technology, Pohang, Korea. \texttt{heekap@postech.ac.kr}}}

\maketitle              

\begin{abstract}
Given an $x$-monotone polygonal chain $T$ with $n$ vertices, 
and an integer $k$,
we consider the problem of finding the lowest horizontal line $L$ lying above $T$ with $k$ point guards lying on $L$, so that
every point on the chain is \emph{visible} from some guard.
A natural optimization is to minimize the $y$-coordinate of $L$.
We present an algorithm for finding the optimal placements of $L$ and $k$ point guards for $T$ in $O(k^2\lambda_{k-1}(n)\log n)$ time for even numbers $k\ge 2$, and in $O(k^2\lambda_{k-2}(n)\log n)$ time for odd numbers $k \ge 3$, where $\lambda_{s}(n)$ is the length of the longest $(n,s)$-Davenport-Schinzel sequence.
We also study a variant with an additional requirement that $T$ is partitioned into $k$ subchains, 
each subchain is paired with exactly one guard, 
and every point on a subchain is visible from its paired guard.
When $L$ is fixed, we can place the minimum number of guards in $O(n)$ time. 
When the number $k$ of guards is fixed, we can find an optimal placement of $L$ with $k$ point guards lying on $L$ in $O(kn)$ time.
\end{abstract}

\section{Introduction}
We consider the problem of guarding a terrain.
A (polygonal) \emph{terrain} is a graph of a piecewise linear function 
$f : A \subset \mathbb{R} \rightarrow \mathbb{R}$ 
that assigns a height $f(p)$ to every point $p$ in a compact domain $A$. 
In other words, a terrain is an $x$-monotone polygonal chain in the plane.
A \emph{guard} is a point lying on or above the terrain.
A point on the terrain is \emph{visible} from a guard if the line segment 
connecting the point and the guard does not cross the terrain.
A subchain of the terrain is visible from a guard 
if every point on the subchain is visible from the guard.
A set of guards \emph{covers} a subchain 
if every point on the subchain is visible from some guard.

Given a terrain $T$ with $n$ vertices, an integer $k$, and a line $L$ parallel to the $x$-axis and lying above $T$, 
we consider the problem of finding $k$ point guards 
lying between $T$ and $L$ 
so that every point on $T$ is visible from some guard.
Observe that for $k\ge \lfloor n/2\rfloor$, 
we can always cover $T$ by placing guards at every other vertex along $T$. 
Thus, we consider the case for $k< \lfloor n/2\rfloor$.
Observe that this problem is closely related to the \textsc{Set Cover} problem: 
the universe is the terrain $T$ and the set corresponding to a guard consists of points on $T$ visible from the guard. 
The goal is to select $k$ sets (or the smallest collection of sets) whose union equals $T$.
Observe that the sets corresponding to $k$ guards are defined after 
the guards are placed. 

\begin{lemma}\label{lem:guard.monotone}
Every point on $T$ visible from a guard $u=(x,y)$ is also visible 
from a guard $u'=(x,y+\delta)$ for any $\delta > 0$.
\end{lemma}
Since the terrain is $x$-monotone, Lemma~\ref{lem:guard.monotone} holds.
By Lemma~\ref{lem:guard.monotone}, a natural optimization is to minimize
the $y$-coordinate (height) of the line $L$.
Indeed, in many real-world applications 
such as security cameras mounted on top of poles, 
it is desirable to make the height of the cameras 
as low as possible so that they are sturdy against high wind 
and it is effective concerning the construction and the cost
while the whole terrain is covered.
We define two problems formally in the following.

\subsubsection*{Altitude Terrain Cover (ATC).}
Given a terrain $T$ and an integer $k$, find the lowest horizontal line $L$ with $k$ point guards lying on $L$ that cover $T$.

When $k=1$, every segment of $T$ must be visible from the unique guard. 
Thus, the optimal placement of $L$ is determined by the intersection
of $n-1$ lines, each extending a segment of $T$. 
Specifically, it is the lowest point in the intersection of 
the (upper) half-planes, each bounded by one of the lines. 
Thus, this can be done in $O(n)$ time~\cite{megiddo1984linear}.
However, for the case $k\ge 2$, the optimal placement of $L$ is not necessarily determined by the intersections of the extending lines.
This is because a segment may not be visible from any single guard
but it is visible from two or more guards in an optimal placement. Fig.~\ref{fig:problems}(a) illustrates a case for $k=2$. 
The optimal line $L$ is not at an intersection of the extending lines.

\subsubsection*{Bijective Altitude Terrain Cover (BATC).}
Given a terrain $T$ and an integer $k$,
find the lowest horizontal line $L$ with $k$ point guards lying on $L$ 
such that $T$ is partitioned into $k$ subchains, 
each subchain is paired with exactly one guard, 
and every point on a subchain is visible from its paired guard. 

We consider two cases. (1) Given a line $L$,
place the minimum number of guards lying between $T$ and $L$ that together cover $T$.
(2) Given an integer $k$, find the lowest horizontal line $L$ and place $k$ guards on $L$ that together cover $T$.
See Fig.~\ref{fig:problems}(b) for an illustration for case (2) with $k=2$.

\begin{figure}[h]
  \centering
  \includegraphics[width=0.8\textwidth]{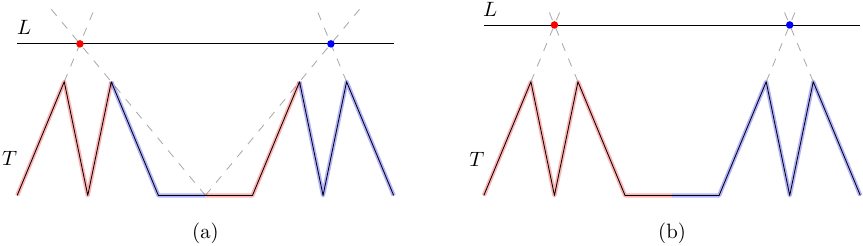}%
  \caption{Cases for $k=2$. (a) The optimal placement of $L$ such that
  $T$ can be covered by two guards lying on $L$. 
  (b) The optimal placement of $L$ such that $T$ is partitioned 
  into two subchains: one of them is visible from the red guard 
  and the other is visible from the blue guard.}
  \label{fig:problems}
\end{figure}

\subsection{Related works}
Daescu et al.~\cite{daescu2019altitude} studied the problem of
placing the minimum number of point guards on a fixed line
so that every point on the terrain is visible from some guard.
They presented a linear time algorithm for the problem.
They also showed that the problem is NP-hard for a polyhedral terrain.

Katoh et al.~\cite{katoh2010parametric} studied the problem of
placing two horizontal segment guards on a fixed line such that the segment guards collectively cover the terrain while minimizing the maximum length of the two segments.
They presented an $O(n\log^2 n)$-time algorithm for the problem.
Later, McCoy et al.~\cite{mccoy2023guarding} studied the problem of finding the lowest horizontal line on which two fixed length horizontal segment guards can be placed to collectively cover the terrain.
They presented an $O(n^2\log n)$-time algorithm for the problem.

There has been work on a variant in which the point guards must be placed 
on the terrain.
Cole and Sharir~\cite{cole1989visibility} showed that 
the problem of covering a polyhedral terrain using the minimum number of guards 
lying on the terrain is NP-hard.
For covering a polygonal terrain using the minimum number of guards 
lying on the terrain, Chen et al.~\cite{chen1995optimal} proposed 
an NP-completeness proof, but it was not completed. 
Later, King and Krohn~\cite{king2011terrain} showed that the problem is NP-hard.

Another line of research on terrain guarding
is to place $k$ vertical line segments on top of the terrain 
such that the terrain is covered by them and 
the maximum length of the segments is minimized.
It is known as the \emph{$k$-watchtower problem}, and 
there is a fair amount of work on the $k$-watchtower problem for a polygonal terrain~\cite{bespamyatnikh2001planar,agarwal2010guarding,seth2023acrophobic,kang4850503guarding} and a polyhedral terrain~\cite{agarwal2010guarding,tripathi2018guarding}.
A variant of the $k$-watchtower problem is to cover the given 
\emph{point sites} lying on the terrain, instead of covering the whole terrain~\cite{kang4850503guarding}.

In the \emph{boundary guarding problem}, we are given a simple polygon 
instead of a terrain. The goal is to minimize the number of point guards 
placed in the polygon to cover all its segments.
Biniaz et al.~\cite{biniaz2024contiguous} gave a polynomial-time exact algorithm 
and an $(\textsf{OPT}+1)$-approximation algorithm under a contiguity restriction, 
where $\textsf{OPT}$ is the number of guards in an optimal solution.

\subsection{Our results}
Let $\lambda_{s}(n)$ denote the length of the longest $(n,s)$-Davenport-Schinzel sequence.
See~\cite{sharir1995davenport} for details about Davenport-Schinzel sequences.
For the Altitude Terrain Cover problem, we give an algorithm
for finding the optimal placements of $L$ and $k$ point guards for $T$
in $O(k^2\lambda_{k-1}(n)\log n)$ time for even numbers $k\ge 2$,
and in $O(k^2\lambda_{k-2}(n)\log n)$ time for odd numbers $k\ge 3$,
where $k$ is the number of guards.
For the Bijective Altitude Terrain Cover problem, we give an algorithm 
for finding the minimum number of guards (case (1)) in $O(n)$ time, and 
an algorithm for finding the lowest horizontal line $L$ (case (2)) in $O(kn)$ time.

\subsubsection*{Sketches of our algorithms.}
For the Altitude Terrain Cover problem, a na\"ive approach would be 
to apply Megiddo's parametric search technique~\cite{megiddo1983applying}.
It enables us to transform an $O(\mathsf{T_s})$-time decision algorithm that performs $O(\mathsf{C_s})$ comparisons into an $O(\mathsf{C_s}\mathsf{T_s})$-time algorithm for its optimization version.
In our problem, a decision algorithm should 
determine whether a given real value is greater than the optimal $y$-coordinate.
Based on $O(n)$-time algorithm by Daescu et al.~\cite{daescu2019altitude},
we can design an $O(kn)$-time decision algorithm that performs $O(n)$ comparisons.
By applying parametric search, we can obtain an $O(kn^2)$-time algorithm.

Our algorithm for the Altitude Terrain Cover problem finds the optimal $y$-coordinate of $L$, denoted by $h^*$, 
and the placements of $k$ guards in an iterative manner. 
It first computes the set of \emph{extreme} placements over 
all $y$-coordinates for the leftmost and the rightmost guards. 
From this, we compute a $y$-interval $I$ that contains $h^*$.
In more detail,
it computes the rightmost placement for the leftmost guard over all $y$-coordinates.
The set of placements forms a piecewise increasing function.
Likewise, it computes the leftmost placement for the rightmost guard over all $y$-coordinates, which together form a piecewise decreasing function.

Then, our algorithm computes a $y$-interval $I$ which contains the optimal $y$-coordinate $h^*$
by binary search over the $y$-coordinates of the breakpoints 
(i.e. the endpoints of the pieces) of the two functions.
Then it repeats the process by computing the functions for 
the second leftmost guard and the second rightmost guard, and
by applying binary search over the $y$-coordinates of the breakpoints 
of the functions within $I$.
Then we obtain a subinterval of $I$ containing $h^*$.

For even numbers $k \geq 2$, we proceed until all guards ($k/2$ guards from the left and $k/2$ guards from the right) are placed. 
Let $I^*$ be the final interval we have, which contains $h^*$.
Let $E$ be the set of edges of $T$ lying between the two guards placed in the middle (their indices being $k/2$ and $k/2 + 1$ from left to right).
Within $I^*$, we check for each edge in $E$ 
whether every point on the edge is visible from any of the two guards, and if not, compute the minimum $y$-coordinate such that the edge is covered together by them. The optimal $y$-coordinate $h^*$ is the maximum $y$-coordinate among the $y$-coordinates.

The time for computing the interval $I^*$ dominates the time complexity of the algorithm, which is $O(k^2\lambda_{k-1}(n)\log n)$.
This is because it takes $O(k)$ time to compute the intersection of two functions that intersect at most $k$ times, the time to compute the two monotone functions (which are upper envelopes of rational functions defined by edges of $T$) is $O(k\lambda_{k-1}(n)\log n)$, and there are $O(k)$ guards.

\begin{comment}

% This takes  time because a pair of functions consists of $k/2$-th guard intersects at most $k-1$ times.
right upper -> $k/2$-th guard f
left upper~ g

$O(k\lambda_{k-1}(n)\log n)$ <- guard 하나마다 <- upper evelope of rational functions such that a pair of functions intersects at most $k-1$
\lambda_{k-1}(n)\log n <- upper evelope of rational intersects at most $k-1$ davernport-shinzel sequence
k <- 위에 시간복잡도는 intersection O(1)시간이라고 가정함

k/2번 반복하면 $O(k^2\lambda_{k-1}(n)\log n)$
\end{comment}

For odd numbers $k$, we continue the process until
we place $k-1$ guards ($\lfloor k/2\rfloor$ guards from the left and $\lfloor k/2\rfloor$ guards from the right).
Instead of computing the edge set $E$ 
and the corresponding $y$-coordinate for each edge in $E$ 
as in the case for even numbers $k$, within $I^*$ we find $h^*$ in $O(kn)$ time using quasiconvex programming.
The total time is again dominated by the time for computing $I^*$, which is $O(k^2\lambda_{k-2}(n)\log n)$.

\bigskip

For the Bijective Altitude Terrain Cover (BATC) problem, 
we first show that there is an optimal partition of $T$ such that
every endpoint of the subchains is a vertex of $T$.
For a subchain, a guard must be placed at the intersection of the upper half-planes, 
each bounded by the line extending a segment of the subchain.
Thus, we compute the lowest point of this intersection to guarantee
the minimum $y$-coordinate for covering the subchain.

In case (1) of BATC, we place guards on $L$ 
one by one from left to right in a greedy manner.
We find the longest subchain starting from the leftmost vertex 
that can be covered by a single guard on $L$ in time linear 
to the size of the subchain. 
We repeat this process for the remaining terrain until the whole terrain is covered.
This can be done in $O(n)$ time.

In case (2) of BATC, there are $C(n-2, k-1)=O((n-2)^{k-1})$ 
partitions of $T$ into $k$ subchains. 
For $k=2$, we can compute the lowest guard position for each subchain
of an optimal partition over all $O(n)$ partitions 
in $O(n)$ time using an incremental algorithm.
For $k>2$, we place $k$ guards one by one from left to right.
We show that it suffices to consider $O(n)$ partitions, 
instead of $O((n-2)^{k-1})$ partitions.
We design an $O(n)$-time incremental algorithm for placing a guard. 
We run the algorithm $k$ times to place $k$ guards, and
obtain the optimal $y$-coordinate of $L$.
Thus, it takes $O(kn)$ time in total.
\bigskip

\section{Preliminaries}
For a point $p$ in the plane, we use $x(p)$ and $y(p)$ to denote 
the $x$- and $y$-coordinate of $p$.
For any two distinct points $p$ and $q$ in the plane, 
we denote by $pq$ the line segment connecting $p$ and $q$,
and by $\overline{pq}$ the line passing through $p$ and $q$.
For a segment $s$, we use $\overline{s}$ to denote the line extending $s$.
The \emph{slope} of a line is the angle swept from the $y$-axis in the counterclockwise direction to the line, and it is thus in $[0,\pi)$. 
The slope of a line segment is the slope of the line that extends it.
For a non-vertical line segment $s$,
we use $s^+$ to denote the set of points in $\mathbb{R}^2$ that lie on or above
$\overline{s}$. Thus, $s^+$ denotes the \emph{upper half-plane} bounded by $\overline{s}$.

We use $T= \langle v_1, \ldots , v_n \rangle$, a sequence of points 
with $x(v_i)<x(v_j)$ for any $1\le i<j\le n$, 
to denote an $x$-monotone polygonal chain which we call \emph{terrain}.
We call every point $v_i$ a \emph{vertex} of $T$ 
for $i=1,\ldots, n$,
and every segment $v_iv_{i+1}$ an \emph{edge} of $T$ for $i=1,\ldots,n-1$.
We use $\tilde{T}$ to denote the set of points $u \in \mathbb{R}^2$ such that the vertical line through $u$ intersects $T$ at $u'$ and
$y(u) \ge y(u')$.
A line segment on $T$ is \emph{fully visible} from a point 
if every point on the line segment is visible from the point. 
We let $\ymax=\max\{y(p)\mid p\in T\}$
denote the maximum $y$-coordinate among the points on $T$.
For a real value $h\in[\ymax,\infty)$, we call the horizontal line $\altitude{h}:=\{(x,y)\in \mathbb{R}^2\vert y=h\}$ as the \emph{altitude line}.

\section{Altitude Terrain Cover}\label{sec:ATC}
In this section, we present algorithms for the Altitude Terrain Cover (ATC) problem.
Consider the case for $k=1$.
An edge $e$ of $T$ is fully visible from a point $u$ only if
$u \in e^+$. Since $T$ is $x$-monotone, every edge of $T$ is fully visible
from the lowest point in $\bigcap_{e \subset T}e^+$. Thus, $h^*$ 
is the $y$-coordinate of the lowest point in $\bigcap_{e \subset T}e^+$.
We can compute the lowest point in
linear time~\cite{megiddo1984linear}.

In Section~\ref{sec:ATC_two}, we give an $O(n\log n)$-time algorithm for $k=2$.
In Section~\ref{sec:ATC_even}, we present an $O(k\lambda_{k-1}(n)\log n)$-time algorithm 
for even numbers $k\ge 2$ of guards. 
In Section~\ref{sec:odd}, we present an $O(k\lambda_{k-2}(n)\log n)$-time algorithm 
for odd numbers $k\ge 3$.

% \subsection{$k = 2$}
\subsection{Two guards.}
\label{sec:ATC_two}
Clearly, $h^*\in\initInterval \coloneq \fullInterval$.
Fix $h \in \initInterval$.
For a point $p\in \tilde{T}$,
let $\fb{p}{h}$ be the rightmost point on $\altitude{h}$ 
from which $p$ is visible.
For a line segment $\ell$ in $\tilde{T}$, 
let $\fb{\ell}{h}$ (resp. $\gb{\ell}{h}$)
be the rightmost (resp. leftmost) point in $\altitude{h}$ 
from which $\ell$ is fully visible. See Fig.~\ref{fig:notations}(a).
Let $\fa{h}$ be the leftmost point among $\fb{e}{h}$ for all edges $e \subset T$. 
See Fig.~\ref{fig:notations}(b).
Similarly, let $\ga{h}$ be the rightmost point among
$\gb{e}{h}$
for all edges $e \subset T$.

Let $\g{1}$ and $\g{2}$ be two guards on $L(h)$ with $x(\g{1}) < x(\g{2})$. 

\begin{figure}[h]
  \centering
  \includegraphics[width=0.8\textwidth]{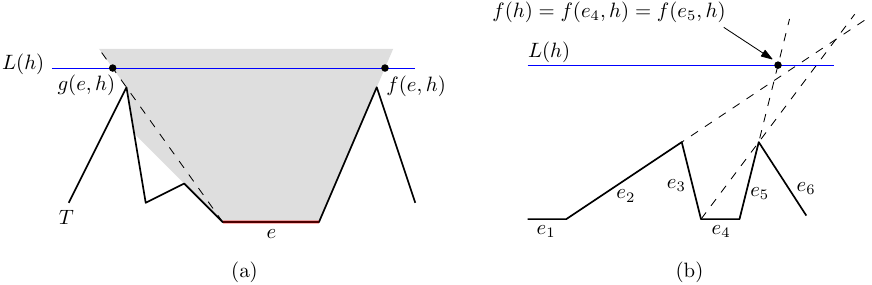}%
  \caption{(a) The gray region is the set of points visible from every point of $e$.
  $\gb{e}{h}$ (resp. $\fb{e}{h}$) is the leftmost (resp. rightmost)  point of $L(h)$ in the region.
  (b) $\fa{h}$ is the leftmost point among $\fb{e}{h}$ for all edges $e$ of $T$.} 
  \label{fig:notations}
\end{figure}

\begin{lemma}\label{lem:ATC.guard_location_Necessary_Condition}
Two guards $\g{1}$ and $\g{2}$ on 
\altitude{h} 
can guard $T$ only if 
$x(\g{1}) \leq x(\fa{h})$ and $x(\g{2}) \geq x(\ga{h})$. 
\label{lem:onlyIf}
\end{lemma}
\begin{proof}
Assume that $x(\fa{h}) < x(\g{1})$. 
Let $e$ be 
an edge of $T$ such that $\fa{h} = \fb{e}{h}$.
Because $x(\g{1}) < x(\g{2})$, there is a point 
on $e$ that is not visible from both $\g{1}$ and $\g{2}$, a contradiction.
Similarly, we can show that 
$x(\g{2}) \geq x(\ga{h})$.
\end{proof}

By Lemma~\ref{lem:ATC.guard_location_Necessary_Condition}, for each $h \in \initInterval$, 
we must place 
$\g{1}$ at $\fa{h}$ or to its left, and
$\g{2}$ at $\ga{h}$ or to its right.
By the following lemma, $\g{1}$ placed at $\fa{h}$ and $\g{2}$ placed at $\ga{h}$ are indeed the best choice.
For a point $p \in \tilde{T}$, let \textsf{Vis}$(T, p)$ be the set of points of $T$ that are visible from $p$.
We simply use $\guarded{p}$ if it is clear from the context.

\begin{lemma}\label{lem:ATC.g_1_location}
$\guarded{u} \subseteq \guarded{\fa{h}}$ for any point $u \in \altitude{h}$ with $x(u) < x(\fa{h})$,
and $\guarded{w} \subseteq \guarded{\ga{h}}$ for any point $w \in \altitude{h}$ with $x(w) > x(\ga{h})$.
\end{lemma}

We need a few technical lemmas.
The first lemma is from Lemma 3 in~\cite{daescu2019altitude}.

\begin{lemma}\label{lem:ATC.visibility_of_left_guards}
  Let $u \in \altitude{h}$ be a point. For any point $p \in T$ 
  with $x(u) < x(p)$,
  if $p$ is not visible from $u$, $p$ is not visible from any point $w\in\altitude{h}$ with $x(w)<x(u)$. 
  For any point $q \in T$ with $x(u) > x(q)$, if $q$ is not visible
from $u$, 
then $q$ is not visible from any point $w\in\altitude{h}$ with $x(w)>x(u)$. 
\end{lemma}

\begin{lemma}\label{lem:ATC.f(h)_left_visible}
Any point $p\in T$ with $x(v_1) \leq x(p) \leq x(\fa{h})$ 
is visible from $\fa{h}$. Any point $q\in T$ with $x(\ga{h}) \leq x(q) \leq x(v_n)$ is visible from $\ga{h}$.
\end{lemma}
\begin{proof}
Let 
$uv$
be an edge of $T$ with $x(u) < x(v)$.
Since $T$ is $x$-monotone, 
$uv$ is visible from any point between $\gb{uv}{h}$ and $\fb{uv}{h}$ on \altitude{h}.
By definition, we have $x(\fa{h}) \le x(\fb{uv}{h})$.
Thus, if $x(u) \le x(\fa{h})$,
then $x(\gb{uv}{h}) \leq x(u) \le x(\fa{h}) \le x(\fb{uv}{h})$, and
$uv$ is fully visible from $\fa{h}$.
Therefore, any point $p \in T$ with $x(v_1) \leq x(p) \leq x(\fa{h})$ is visible from $\fa{h}$. The second claim also holds analogously.
\end{proof}

Now we are ready to prove Lemma~\ref{lem:ATC.g_1_location}.

\begin{proof}
\textit{Proof of Lemma~\ref{lem:ATC.g_1_location}}.
By Lemma~\ref{lem:ATC.f(h)_left_visible}, for any point $p\in T$ with 
$x(v_1) \leq x(p) \leq x(\fa{h})$ is visible from $\fa{h}$.
For any point $p\in T$ with $x(p) > x(\fa{h})$,
if $p$ is not visible from $\fa{h}$, 
$p$ is not visible from $u$ by Lemma~\ref{lem:ATC.visibility_of_left_guards}.
Thus, for any point $u \in \altitude{h}$ with $x(u) < x(\fa{h})$, $\guarded{u} \subseteq \guarded{\fa{h}}$.
The second claim also holds analogously.
\end{proof}

By Lemma~\ref{lem:ATC.g_1_location},
for a fixed $h$,
we find $\fa{h}$ and $\ga{h}$, and then check whether they together cover $T$.
Using this, we find the minimum value of $h$ such that $T$ is covered.
To find it efficiently, 
we show that both $\fFunc$ and $\gFunc$ are piecewise linear, 
and that they are $xy$-monotone curves in $\mathbb{R}^2$.

A \emph{shortest path tree} of a polygon 
from a point $p$ in the polygon is the tree rooted at $p$
such that for each vertex $u$ of the polygon, 
the path between $p$ and $u$
on the graph (as a geometric graph) coincides with the geodesic between $p$ and $u$ in the polygon~\cite{guibas1986linear}.
See Fig.~\ref{fig:strong_visibility_line}(a).
It can be computed in linear time for any given point in the polygon. 
For a point $p \in \tilde{T}$, let $\SPT{p}$ be the shortest path tree from $p$ to the vertices of $T$. Let $\mathsf{S}_1=\SPT{v_1}$ and $\mathsf{S}_n=\SPT{v_n}$.
For a node $v$ in $\mathsf{S}_t$, let $\pi_t(v)$ denote the parent node of $v$
in $\mathsf{S}_t$ for $t=1,n$.

\begin{lemma}\label{lem:ATC.vertex_visibility_SPT}
Let $v$ be a vertex of $T$.
For any $h \in \initInterval$,
$\fb{v}{h} = \altitude{h} \cap \overline{v\pi_n(v)}$, and
$\gb{v}{h} = \altitude{h} \cap \overline{v\pi_1(v)}$.
\end{lemma}

\begin{proof}
Let $u=\altitude{h} \cap \overline{v\pi_n(v)}$.
By the definition of the shortest path tree, 
every point $q \in T$ with $x(v) < x(q)$
lies below or on $\overline{v\pi_n(v)}$.
Therefore, $v$ is visible from $u$.
For any $w\in\altitude{h}$ with $x(w)>x(u)$, $v$ is not visible from $w$ because of the edge incident to $\pi_n(v)$.
The second claim also holds analogously.
\end{proof}

\begin{figure}[h]
  \centering
  \includegraphics[width=0.8\textwidth]{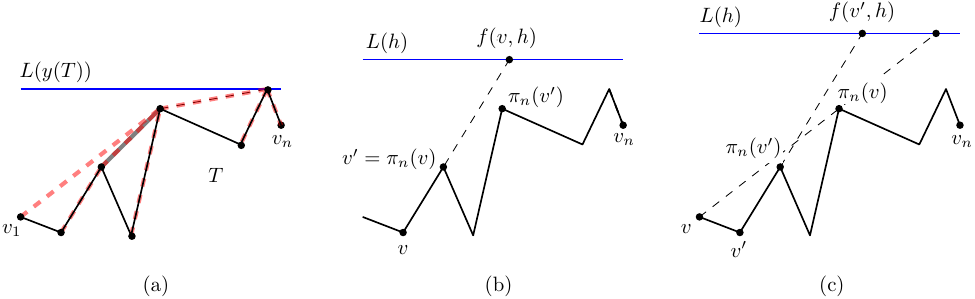}%
  \caption{(a) Red dashed segments are edges of $\mathsf{S}_n$. 
  For any fixed $h \in \initInterval$, 
  (b) $\fb{vv'}{h} = \fb{v}{h}$ if $x(\pi_n(v)) < x(\pi_n(v'))$. 
  (c) $\fb{vv'}{h} = \fb{v'}{h}$ if $x(\pi_n(v)) \ge x(\pi_n(v')$.}
  \label{fig:strong_visibility_line}
\end{figure}

\begin{lemma}\label{lem:ATC.edge_strong_visibility_case}
    Let $vv'$ be an edge of 
    $T$ with $x(v) < x(v')$.
    For any $h \in \initInterval$, $\fb{vv'}{h} = \fb{v}{h}$ if $x(\pi_n(v)) < x(\pi_n(v'))$, and $\fb{vv'}{h} = \fb{v'}{h}$ otherwise.
\end{lemma}
\begin{proof}
Observe that either $\fa{vv',h} = \fa{v,h}$ or $\fa{vv',h} = \fa{v',h}$
because $T$ is $x$-monotone.
Consider the case $x(\pi_n(v)) < x(\pi_n(v'))$.
Suppose $\pi_n(v)\neq v'$. Then
$x(v') < x(\pi_n(v))$ and $v'$ lies below $\overline{v\pi_n(v)}$.
Since $\pi_n(v)$ and $\pi_n(v')$ are vertices of $\mathsf{S}_n$, $\pi_n(v')$ lies on or below 
$\overline{v\pi_n(v)}$.
Then $\pi_n(v')$ is not visible from $v'$, a contradiction.
Thus, $\pi_n(v)=v'$ and $\fb{vv'}{h} = \fb{v}{h}$.
See Fig.~\ref{fig:strong_visibility_line}(b).

Next, consider the case $x(\pi_n(v)) \geq x(\pi_n(v'))$.
Then $\pi_n(v)$ lies on or below $\overline{v'\pi_n(v')}$ 
and $\pi_n(v')$ lies on or below $\overline{v\pi_n(v)}$.
Thus, $v\pi_n(v)$ has slope smaller than
the slope of $v'\pi_n(v')$.
Since $\altitude{h} \cap \overline{v\pi_n(v)}$ lies right to $\altitude{h} \cap \overline{v'\pi_n(v')}$ we have $\fb{vv'}{h} = \fb{v'}{h}$.
See Fig.~\ref{fig:strong_visibility_line}(c).
\end{proof}

Having computed $\SPT{v_1}$ and $\SPT{v_n}$ in linear time, we can compute $\fb{e}{h}$ and $\gb{e}{h}$ for any edge $e$ of $T$ and any fixed 
$h \in \initInterval$ in $O(1)$ time 
by Lemmas~\ref{lem:ATC.vertex_visibility_SPT} and~\ref{lem:ATC.edge_strong_visibility_case}.

\begin{lemma}\label{lem:ATC.f(h)_properties_time_complexity}
$\fFunc$ is piecewise linear and monotone increasing, and
$\gFunc$ is piecewise linear and monotone decreasing. 
We can compute them in $O(n)$ time.
\end{lemma}
\begin{proof}
By definition, $f(h)$ is equal to the upper envelope of 
$\fb{e}{h}$'s for all edges $e$ of $T$.
See Fig.~\ref{fig:ATC_first_interval}(a).
By Lemmas~\ref{lem:ATC.vertex_visibility_SPT} and~\ref{lem:ATC.edge_strong_visibility_case}, 
for any edge $e$ of $T$,
$\fb{e}{h}$ is an increasing linear function.
For any fixed $h \in \initInterval$,
we can compute $\fb{e}{h}$ in $O(1)$ time.
Thus, we can compute $f(h)$ in $O(n)$ time~\cite{seth2023acrophobic}.
Analogously, $\gFunc$ is piecewise linear and monotone decreasing,
and we can compute it in $O(n)$ time.
\end{proof}

\begin{figure}[h]
  \centering
  \includegraphics[width=0.75\textwidth]{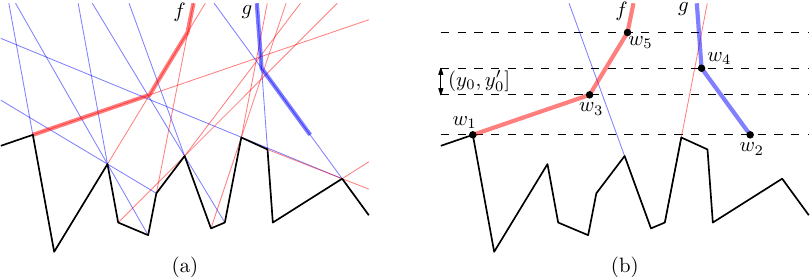}%
  \caption{(a) Red and blue half-line represents $\fb{e}{h}$ and $\gb{e}{h}$ for an edge $e$ of $T$, respectively. 
  (b) We have $W = \langle w_1, w_2, w_3, w_4, w_5 \rangle$, and
  $(y_0, y'_0]=\big(y(w_3),y(w_4)\big]$.}
  \label{fig:ATC_first_interval}
\end{figure}

By Lemma~\ref{lem:ATC.f(h)_properties_time_complexity},
we compute $\fFunc$ and $\gFunc$ in $O(n)$ time.
The next step is to reduce the search space for $h^*$ to the maximal subinterval $(y_0, y'_0]$ of $\initInterval$ such that $h^* \in (y_0, y'_0]$, and both $\fFunc$ and $\gFunc$ restricted to $(y_0, y'_0]$ are linear functions.

The following lemma is from Theorem 3 in~\cite{daescu2019altitude}.

\begin{lemma} % [Theorem 3 in~\cite{daescu2019altitude}]
\label{lem:ATC.comparison_time}
For a fixed $h\in \initInterval$, we can compute a 
minimum-sized set of guards 
to cover $T$ in $O(n)$ time.
\end{lemma}

Let $W=\langle w_1,\ldots,w_m\rangle$ be the 
sequence of vertices of $\fFunc$ and $\gFunc$, 
ordered by their $y$-coordinates, where $m=O(n)$.
Let $H$ denote the sequence 
$\langle y(w_1),y(w_2),\ldots, y(w_m)\rangle$.

The following lemma, Lemma~\ref{lem:ATC.first_binary_search}, enables us to find $(y_0, y'_0]$ such that $y_0$ and $y'_0$ are consecutive in $H$ using binary search.
See Fig.~\ref{fig:ATC_first_interval}(b).

\begin{lemma}\label{lem:ATC.first_binary_search}
We can compute the interval $(y_0, y'_0]$ in $O(n \log n)$ time.
\end{lemma}
\begin{proof}
Let $w_i$ be the lowest vertex in $W$ such that $T$ can be covered by two guards lying on $L(y(w_i))$.
It is clear that
$(y_0, y'_0] = \big(y(w_{i-1}),y(w_i))\big]$.
By Lemma~\ref{lem:ATC.f(h)_properties_time_complexity}, 
we can compute $\fFunc$ and $\gFunc$ in $O(n)$ time.
Since $m=O(n)$, we can compute $w_i$ in $O(n \log n)$ time 
by binary search using Lemma~\ref{lem:ATC.comparison_time}.
\end{proof}

Let $q \in \tilde{T}$ and $e$ be an edge of $T$ that is not fully visible from $q$.
Let $e'$ be the maximal segment of $e$
that is fully visible from $q$.
Let $p$ be the endpoint of $e'$ that is not an endpoint of $e$.
The \emph{peak} of $e$ for $q$, denoted by $\peak{e}{q}$,
is the set of vertices $v \in T$ such that
$v \in pq$.
See Fig.~\ref{fig:change_events}(b).
For edges $e$ with $e'=\emptyset$, we set $\peak{e}{q}=\emptyset$.

Recall that $\g{1}$ and $\g{2}$ are the two guards.
Let $\fullLeft{}$ be the set of fully visible edges from $\g{1}$ on $L(h)$.
Let $\partialLeft{}$ be the set of pairs $(e, \peak{e}{\g{1}})$ for edges $e$ that are not fully visible from $\g{1}$ on $L(h)$, and $\peak{e}{\g{1}} \neq \emptyset$.
We similarly define $\fullRight{}$ and $\partialRight{}$ 
for $\g{2}$ on $L(h)$.
Let $(y_1,y'_1)$ be the maximal subinterval of $(y_0, y'_0]$
such that $h^* \in (y_1,y'_1]$, and 
the sets $\fullLeft{}$, $\partialLeft{}$, $\fullRight{}$, and $\partialRight{}$ remain unchanged for all $h\in (y_1,y'_1)$.

\begin{figure}[h]
  \centering
  \includegraphics[width=0.75\textwidth]{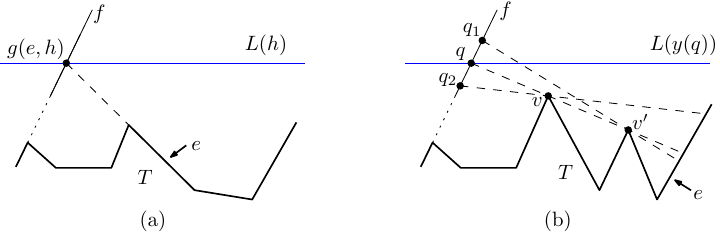}%
  \caption{  (a) $e$ is an edge of $T$. $f(h)=g(e,h)$.
  (b) $vv'$ is an edge of $\mathsf{S}_1$.
  We have that 
  $\overline{vv'} \cap e$ is contained in the interior of $e$, $q = \fFunc \cap \overline{vv'}$, and 
  $\peak{e}{q} = \{v, v'\}$ for $h = y(q)$.
  For arbitrary small $\epsilon > 0$,  
  $\peak{e}{q_1} = \{v'\}$ for $h = y(q) + \epsilon$, and
  $\peak{e}{q_2} = \{v\}$ for $h = y(q) - \epsilon$.}
    \label{fig:change_events}
\end{figure}

\begin{lemma}\label{lem:ATC.second_binary_search}
We can compute the interval $(y_1,y'_1)$ in $O(n\log n)$ time.
\end{lemma}
\begin{proof}
Let $E$ be the set of edges of $\mathsf{S}_1$ with negative slopes. Let $p_1, \ldots, p_m$ be the intersection points between $f$ and $\overline{e}$ for all edges $e \in E$,
sorted in increasing order of $y$-coordinate.
Since $f$ is monotone increasing by Lemma~\ref{lem:ATC.f(h)_properties_time_complexity},
we have $m=O(n)$.
We can compute the intersection points and sort them in $O(n\log n)$ time.

We show that on any $y$-interval $h \in (y(p_i), y(p_{i+1}))$
for $i = 1, \ldots, m-1$, the sets $\fullLeft{}$ and $\partialLeft{}$ remain the same.
It is clear that $\fullLeft{}$ remains the same on such a $y$-interval
by Lemmas~\ref{lem:ATC.vertex_visibility_SPT} and~\ref{lem:ATC.edge_strong_visibility_case}, 
because the ray $\gb{e}{h}$ is contained in the line extended from an edge of $\mathsf{S}_1$,
for any fixed edge $e$ of $T$ and $h \in \initInterval$.
See Fig.~\ref{fig:change_events}(a).

Now consider the set $\partialLeft{}$.
For two distinct vertices $v, v'$ of $T$, 
let $e$ be the edge of $T$ that the ray emanating from $v$ 
towards $v'$ meets for the first time.
If the ray intersects $f$ at $q$ before hitting $T$, 
the peak of $e$ changes at $y(q)$.
See Fig.~\ref{fig:change_events}(b).
Observe that 
if $vv'$ is not an edge of $\mathsf{S}_1$, 
the ray emanating from $v'$ towards $v$ intersects $T$ 
before hitting $f$.
Thus, for any $y$-interval $h \in (y(p_i), y(p_{i+1}))$
for $i = 1, \ldots, m-1$, $\partialLeft{}$ remains the same.

Similarly, for the set $E'$ of edges of $\mathsf{S}_n$ with positive slopes and the intersection points $q_1, \ldots, q_{m'}$ of $g$ and $\overline{e}$ for all edges $e \in E'$, we can show that on any $y$-interval $h \in (y(q_i), y(q_{i+1}))$
for $i = 1, \ldots, m'-1$, $\fullRight{}$ and $\partialRight{}$ remain the same.

To find $(y_1,y'_1)$ by binary search, we compute $p_1,\ldots, p_m$ and $q_1, \ldots, q_{m'}$ and sort them in $O(n\log n)$ time.
Since $m=O(n)$ and $m'=O(n)$, the number of binary search steps is $O(\log n)$.
By Lemma~\ref{lem:ATC.comparison_time}, each comparison of the binary search is done in linear time.
Thus, we can compute $(y_1,y'_1)$ in $O(n\log n)$ time.
\end{proof}

By Lemma~\ref{lem:ATC.second_binary_search}, we can compute $(y_1,y'_1)$ in $O(n\log n)$ time.
For each edge $e$ of $T$ such that
$e \notin U(y_1)$ and $e \notin \overline{U}(y_1)$,
we compute the minimum $y$-coordinate $h$
such that $e$ is covered by $\g{1}$ and $\g{2}$ on $L(h)$ in $O(1)$ time. 
The optimal $y$-coordinate $h^*$ is the maximum among 
the $y$-coordinates.
Thus, we have 
the following theorem.

\begin{theorem}\label{thm:ATC_two}
  Given a terrain with $n$ vertices, the Altitude Terrain Cover problem for $k=2$ can be solved in $O(n \log n)$ time.
\end{theorem}
\begin{proof}
First, we determine for every edge $e$ of $T$ 
such that $e \notin U(y_1)$ and $e \notin \overline{U}(y_1)$
in $O(n)$ time.
We can compute $\mathsf{S}_1$ and $\mathsf{S}_n$ in linear time~\cite{guibas1986linear}.
With the shortest path trees, we can compute $\fb{e}{y_1}$ and $\gb{e}{y_1}$ in $O(1)$ time for each edge $e$ of $T$ by Lemmas~\ref{lem:ATC.vertex_visibility_SPT} and~\ref{lem:ATC.edge_strong_visibility_case}.
If $\gb{e}{y_1} \leq f(y_1) \leq \fb{e}{y_1}$, $e$ is fully visible from $u_1$.
If $\gb{e}{y_1} \leq g(y_1) \leq \fb{e}{y_1}$, $e$ is fully visible from $u_2$.

For every edge $e$ of $T$ that is not fully visible from $\g{1}$,
the maximum portion of $e$ visible from $\g{1}$
can be computed in $O(n)$ time by Lemma 8 in~\cite{daescu2019altitude}.
During the computation of the visible portion for each edge $e$ of $T$ at $y_1$, we also compute $\peak{e}{\g{1}}$ at $y_1$. 
Similarly, for each edge $e$ of $T$, we can compute $\peak{e}{\g{2}}$ at $y_1$. 
Finally, for each edge $e$ of $T$ that is not fully visible, we can compute the minimum $h$ such that $e$ is covered by $\g{1}$ and $\g{2}$ on $L(h)$ in $O(1)$ time with $\peak{e}{\g{1}}$ and $\peak{e}{\g{2}}$.
\end{proof}

\subsection{Even numbers of guards.}\label{sec:ATC_even}
Here we consider even numbers $k > 2$.
Let $\g{1}, \ldots, \g{k}$ be the guards on $L(h)$
from left to right.
Our algorithm works in a greedy manner.
It places $\g{1}$ to cover a subchain containing $v_1$, and places $\g{k}$ to cover a subchain containing $v_n$. Let $\mathcal{T}_1$ and $\mathcal{T}_k$ be the set of maximal subchains of $T$ covered by $\g{1}$ and $\g{k}$, respectively.
\begin{comment}
It places $\g{1}$ at $\fa{h}$ to cover 
a set $\mathcal{T}_1$ of maximal subchains of $T$,
one of which contains $v_1$.
% the maximal \ccheck{subchains} $T_1$ of 
% containing $v_1$, and
Similarly, it places $\g{k}$ at $\ga{h}$ to cover a set $\mathcal{T}_k$ of maximal subchains of $T$, one of which contains $v_n$.
% the maximal 
% \ccheck{subchains} $T_k$ of $T$ containing $v_n$.
\end{comment}
The algorithm then places $\g{2}$ (and $\g{k-1}$) to cover 
the maximal subchains of $T\setminus (\mathcal{T}_1\cup \mathcal{T}_k)$ containing 
the leftmost (and the rightmost) vertex of it.
It continues to place the remaining guards in this
manner until it places $\g{k/2}$ and $\g{k/2+1}$.
The following lemma and corollary generalize  
Lemma~\ref{lem:ATC.guard_location_Necessary_Condition}, which gives a necessary condition on the position of guards.

\begin{lemma}\label{lem:ATC_even.guard_location_necessary_recursive}
Suppose that $\g{1},\ldots,\g{i}$ and $\g{k-i+1},\ldots,\g{k}$ 
are placed on $L(h)$ for $1 \le i< k/2$.
Then $T$ can be covered only if both of the followings hold.
\begin{denseitems}
	\item[\textnormal{(a)}] $x(\g{i+1}) \leq x(p)$, where $p$ is the leftmost point among $\fb{e}{h}$'s for all edges $e$ of $T\setminus \bigcup^{i}_{j=1} \big(\guarded{\g{j}} \cup\guarded{\g{k-j+1}}\big)$.
	\item[\textnormal{(b)}] $x(q)\le x(\g{k-i})$, where $q$ is 
the rightmost point among $\gb{e}{h}$'s for all edges $e\subset T\setminus \bigcup^{i}_{j=1} \big(\guarded{\g{j}} \cup\guarded{\g{k-j+1}}\big)$.
\end{denseitems}
\end{lemma}
\begin{proof}
Assume to the contrary that $x(p) < x(\g{i+1})$.
Let $e$ be an edge of $T \setminus \bigcup^{i}_{j=1} \big(\guarded{\g{j}} \cup\guarded{\g{k-j+1}}\big)$ such that $\fb{e}{h} = p$.
Using an argument similar to the proof of Lemma~\ref{lem:ATC.guard_location_Necessary_Condition},
we can show that there is a point on $e$ that is not visible from $\g{i+1},\ldots,\g{k-i}$ regardless of their positions.
Similarly, we can show that $x(q) \leq x(\g{k-i}$).
\end{proof}

Let $\fFuncI{1}=\fFunc$, $\gFuncI{1}=\gFunc$ and $\unseen{1} = T$.
Let $\unseen{i+1} = \unseen{i}\setminus \big(\guarded{\fFuncI{i}}\cup \guarded{\gFuncI{i}} \big)$
be the parts of $T$ not covered by any of 
$\fFuncI{1}$, \ldots , $\fFuncI{i}$ and $\gFuncI{1}$, \ldots, $\gFuncI{i}$.
We define $\fFuncI{i}$ 
as the leftmost point among $\fb{e}{h}$ for all edges $e$ of $\unseen{i}$.
Similarly, we define $\gFuncI{i}$ 
as the rightmost point among $\gb{e}{h}$ for all edges $e$ of $\unseen{i}$.

With Lemma~\ref{lem:ATC_even.guard_location_necessary_recursive}, we place the $k$ guards on $\altitude{h}$ such that $\g{i} = f_i(h)$ for $i \leq k/2$, and $\g{i} = g_{k+1-i}(h)$ for $k/2 < i \leq k$.

\begin{lemma}\label{lem:ATC_even.guard_location_greedy_recursive}
Let $i$ be an integer with $1 < i\leq k/2$. For any point $u$ on $\altitude{h}$ with $x(\fFuncI{i-1}) \leq x(u) \leq x(\fFuncI{i})$,
$\unseen{i}\setminus \guarded{\fFuncI{i}} \subseteq \unseen{i} \setminus \guarded{u}$. For any point $u'$ on $\altitude{h}$ with $x(\gFuncI{i}) \leq x(u') \leq x(\gFuncI{i-1})$, 
$\unseen{i}\setminus \guarded{\gFuncI{i}} \subseteq \unseen{i} \setminus \guarded{u'}$.
\end{lemma}
\begin{proof}
By Lemma~\ref{lem:ATC.f(h)_left_visible}, any point
$q\in \unseen{i}$
with $x(q) \leq x(\fFuncI{i})$ is visible from $\fFuncI{i}$.
For any point
$q\in \unseen{i}$
with $x(q) \geq x(\fFuncI{i})$, if $q$ is not visible from $\fFuncI{i}$, $q$ is not visible from $u$ by Lemma~\ref{lem:ATC.visibility_of_left_guards}. Thus the lemma holds. The second claim also holds analogously.
\end{proof}

Lemma~\ref{lem:ATC_even.guard_location_greedy_recursive}, which is a generalization of Lemma~\ref{lem:ATC.g_1_location}, indicates that
this placement is indeed the best choice.
For an edge $e\in \unseen{i}$, we show that $\fb{e}{h}$ is a piecewise rational function for any $i=2,\ldots,k/2$.
For a polynomial function $\phi$, let $\deg(\phi)$ be the degree of $\phi$.
For two integers $i \geq 0$ and $j > 0$, 
let 
$\SRF{i}{j}$ be the set of rational functions $C(h)/D(h)$, where $C(h)$ and $D(h)$ are polynomials with $\deg(C(h))\leq i$ and $\deg(D(h))\leq j$.

Let \full{i}
be the set of edges $e$ of $T$ such that every point of $e$ is visible from some guard among $\g{1}\ldots,\g{i}$ and $\g{k-i+1}, \ldots, \g{k}$
on $\altitude{h}$. 
Let $\partialLeft{i}$ be the set of pairs $(e, \peak{e}{\g{i}})$ 
for edges $e$ of $T$ such that
$e \notin \full{i}$,
for $1 \leq i \leq k/2$.
Similarly, let $\partialRight{i}$ be the set of pairs $(e, \peak{e}{\g{k-i+1}})$ for edges $e$ of $T$
such that
$e \notin \full{i}$,
for $1 \leq i \leq k/2$.

Let $\interval{i} = (h_i, h'_i)$ for $i=1,\ldots, k/2$ be the maximal subinterval of $\interval{i-1}$ such that $h^* \in (h_i,h'_i]$, the sets
$\partialLeft{i}, \partialRight{i}$, and $\full{i}$
remain the same for all $h \in \interval{i}$, and both $\fFuncI{i}$ and $\gFuncI{i}$ restricted to $\interval{i}$ are rational functions in $\SRF{i}{i-1}$.

For an edge $e$ of $T$ and an integer $i$ with $1<i\leq k/2$,
let $e(i,h)$ denote the portion of $e$ that is not visible from $\fFuncI{i-1}$ or $\gFuncI{i-1}$ on $\altitude{h}$, and 
let $\interval{i}(e) \subseteq \interval{i-1}$ be a $y$-interval such that
$h^*\in \interval{i}(e)$ and
$\peak{e}{\fb{e(i,h)}{h}}$
remains the same for all
$h\in \interval{i}(e)$. 

\begin{lemma}\label{lem:ATC_even.rational_function_degree}
For an integer $i$ with $1<i\leq k$, 
a $y$-coordinate $h \in \interval{i}(e)$, 
and an edge $e$
of $T$ with $e \notin\fullLeft{i-1}$,
$\fb{e(i,h)}{h}$
is a function in $\SRF{i}{i-1}$.
\end{lemma}
\begin{proof}
Let $vv'=e$.
Without loss of generality, assume $x(v) \leq x(v')$.
Let $pv'$ be the maximum portion of $vv'$ visible from $\fFuncI{i-1}$.
Let $vq$ be the maximum portion of $vv'$ visible from $\gFuncI{i-1}$.
By Lemma~\ref{lem:ATC.edge_strong_visibility_case}, $\fb{vv'(i,h)}{h}$ is either $\fb{vp}{h}$ or $\fb{qv'}{h}$.
Without loss of generality, assume $\fb{vv'(i,h)}{h}=\fb{vp}{h}$.

See Fig.~\ref{fig:rational_function}.
Let $\ell_1$ be the line through points
$\fFuncI{i-1} = \big(C(h)/D(h), h\big)$, where $C(h)$ and $D(h)$ are polynomial functions.
$\peak{vv'}{\fFuncI{i-1}} = (x_1,y_1)$, that is,

\[\ell_1 \coloneq y=\dfrac{(h-y_1)D(h)}{C(h)-x_1D(h)}x+\dfrac{y_1C(h) - x_1hD(h)}{C(h)-x_1D(h)}.\]
Let $p$ be the intersection point between $\ell_1$ and $\overline{vv'}\coloneq y=m_1x+b_1$, that is, the point at
\[\Bigg(\dfrac{(y_1-b_1)C(h)+(b_1x_1-x_1h)D(h)}{m_1C(h)+(y_1-h-m_1x_1)D(h)},\dfrac{m_1y_1C(h)+(b_1y_1-b_1h-x_1m_1h)D(h)}{m_1C(h)+(y_1-h-m_1x_1)D(h)}\Bigg).\]

Let $\alpha=m_1C(h)+(y_1-h-m_1x_1)D(h)$, $\beta=(y_1-b_1)C(h)+(b_1x_1-x_1h)D(h)$ and $\gamma=m_1y_1C(h)+(b_1y_1-b_1h-x_1m_1h)D(h)$.
Let $\ell_2$ be the line through $p$ and
$\peak{vv'}{f(vv'(i,h),h)}$, that is,
\[\ell_2:=y=\dfrac{\gamma-y_2\alpha}{\beta-x_2\alpha}x+\dfrac{y_2\beta-x_2\gamma}{\beta-x_2\alpha}.\]

Substituting $h$ into $y$ in the above equation yields the following:
\[\fb{vv'(i,h)}{h} = \Bigg(\dfrac{(\beta-x_2\alpha)h+(x_2\gamma-y_2\beta)}{\gamma-y_2\alpha},h\Bigg).\]

Let $a = \max\{\deg(\alpha),\deg(\beta),\deg(\gamma)\} \leq \max\{\deg(C(h)),\deg(D(h))\}$.
Then $\bigcup_{h \in \interval{i}}\fb{vv'(i,h)}{h}$ represents a function in $\SRF{a+1}{a}$.
When $i=2$, we have $a=1$ because $\fFuncI{1}$ is a linear function when restricted to $\interval{1}$.
In the same way, we have $a = i-1$ for every $1 < i \leq k$, and $\bigcup_{h \in \interval{i}(vv')}\fb{vv'(i,h)}{h}$ represents a rational function in $\SRF{i}{i-1}$. 
\end{proof}

\begin{figure}[h]
  \centering
  \includegraphics[width=0.7\textwidth]{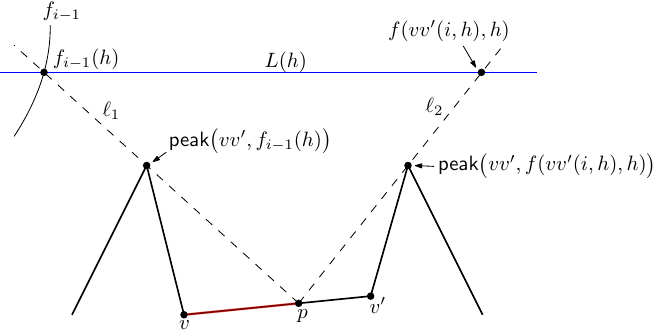}%
  \caption{Proof of Lemma~\ref{lem:ATC_even.rational_function_degree}.}\label{fig:rational_function}
\end{figure}

Lemma~\ref{lem:ATC_even.rational_function_degree} shows that $\fb{e}{h}$ restricted to $\interval{i}(e)$
is
a rational function. 
The following two lemmas show how to compute $I_1$
and $\interval{i}$'s for $i = 2, \ldots, k/2$ iteratively.
\begin{lemma}\label{lem:ATC_even.first_interval}
We can compute $\interval{1}$ in $O(n\log n)$ time.
\end{lemma}
\begin{proof}
By Lemma~\ref{lem:ATC.second_binary_search}, we can find the interval $(y_1,y'_1)$, 
such that $U(y_1)$, $\overline{U}(y_1)$, $V_1(y_1)$, and $\overline{V}_1(y_1)$ remain unchanged for all $h\in (y_1,y'_1)$, in $O(n\log n)$ time.
Therefore, we consider the maximal subinterval of $(y_1,y'_1)$ such that $U_1(h)$ remains the same for all $h$ in the maximal subinterval.
For each edge $e$ of $T$ such that
$e \notin U(y_1)$ and $e \notin \overline{U}(y_1)$,
we compute the minimum $y$-coordinate $h$
such that $e$ is covered by $\g{1}$ and $\g{2}$ together on $L(h)$ in $O(1)$ time using an argument similar to the proof of Theorem~\ref{thm:ATC_two}.
We compute such $h$ for each edge, sort them, and compute $\interval{1}$ using binary search by Lemma~\ref{lem:ATC.comparison_time}. This takes $O(n\log n)$ time.
\end{proof}

After computing $I_1$ by Lemma~\ref{lem:ATC_even.first_interval}, we compute $I_i$ for $i=2,\ldots, k/2$ in an iterative manner.
\begin{lemma}\label{lem:ATC_even.compute_strong_visibility_each_edge}
For any integer $i$ with $1< i \leq k/2$,
we can compute $\interval{i}$ given $\interval{i-1}$ in $O(k\lambda_{2i-1}(n)\log n)$ time.
\end{lemma}
\begin{proof}
We first compute the $y$-interval $I$ such that $\interval{i} \subseteq I \subseteq \interval{i-1}$, and $\fb{e}{h}$ and $\gb{e}{h}$ for $h\in I$ are in $\SRF{i}{i-1}$ for any edge $e$ of $T$ with $e\not\in U_{i-1}(h_{i-1})$. 

By Lemma~\ref{lem:ATC_even.rational_function_degree}, for any edge $e$ of $T$ with $e\not\in U_{i-1}(h_{i-1})$, $\fb{e}{h}$ for $h\in \interval{i}(e)$ represents a function in $\SRF{i}{i-1}$.
We compute $I_i$ which is the intersection of all $I_i(e)$'s.

For an edge $e$, let $p_e(vv')$ be defined as follows: if the ray emanating from $v$ toward $v'$ 
intersects $T$ at a point $q \in e$ for the first time, $p_e(vv') = q$. If not, $p_e(vv') = \varnothing$.
Let $P_e$ be the set of $p_e(vv')$'s for all edges $vv'$ (with $x(v) < x(v')$) of $\mathsf{S}_n$.

Let $H$ be the set of $h \in \interval{i-1}$ such that the ray emanating from $\fFuncI{i-1}$ and passing through $\peak{e}{\fFuncI{i-1}}$ intersects $e$ at a point in $P_e$ before hitting any other edge of $T$. 
We can compute $P_e$ in $O(n)$ time by traversing $\mathsf{S}_n$.
See Fig.~\ref{fig:shortest_path_tree_property} for an illustration.
After computing $P_e$, we can compute $H$ in $O(n)$ time by Lemma 7 in~\cite{daescu2019altitude} and sort it in $O(n\log n)$ time.

We find the interval $I$ defined by consecutive values of $H$ in $O(n\log n)$ time using binary search by Lemma~\ref{lem:ATC.comparison_time}.
We then compute $\fFuncI{i}$ restricted to $I$.
Since any pair of rational functions in $\SRF{i}{i-1}$ intersect at most $2i-1$ times and each function is fully defined in $I$,
we can compute $\fFuncI{i}$ restricted to $I$ in $O(k\lambda_{2i-1}(n)\log n)$ time and its complexity is $O(\lambda_{2i-1}(n))$~\cite{sharir1995davenport}.

Next, we compute $I'\subseteq I$ such that $h^*\in I'$ and both $\fFuncI{i}$ and $\gFuncI{i}$, when restricted to $I'$, are rational functions in $\SRF{i}{i-1}$.
Since the complexities of $\fFuncI{i}$ and $\gFuncI{i}$ are $O(\lambda_{2i-1}(n))$, we can compute $I'$ in $O(\lambda_{2i-1}(n)\log n)$ time using binary search by Lemma~\ref{lem:ATC.comparison_time}.

Finally, we compute $\interval{i}\subseteq I'$ in $O(n\log n)$ time such that $\full{i}$, $\partialLeft{i}$ and $\partialRight{i}$ remain unchanged in $\interval{i}$, 
using an argument similar to the proof of Lemmas~\ref{lem:ATC.second_binary_search} and~\ref{lem:ATC_even.first_interval}.
\end{proof}

\begin{figure}[h]
  \centering
  \includegraphics[width=0.7\textwidth]{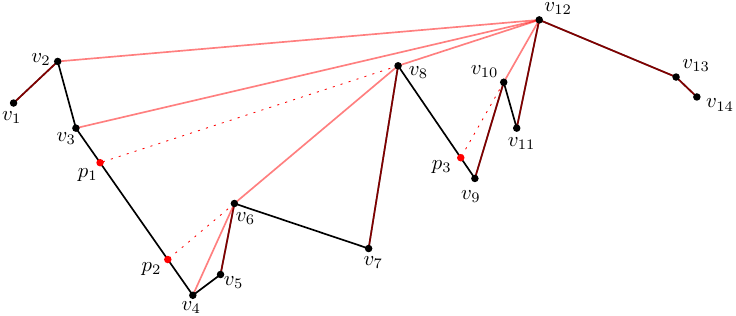}%
  \caption{The ray emanating from $v_{12}$ towards $v_8$ intersects $v_3v_4$ at $p_1$. 
  The ray emanating from $v_{8}$ towards $v_{6}$ intersects $v_3v_4$ 
  at $p_2$.
  The ray emanating from $v_{12}$ towards $v_{10}$ intersects $v_8v_9$ at
 $p_3$. Since $\pi_n(v_6)=v_8$, $x(p_1)\leq x(p_2)$.
 Since $\pi_n(\pi_n(v_6))=\pi_n(v_{10})$ and $x(\pi_n(v_6)=v_8) \leq x(v_{10})$, $x(p_2)\leq x(p_3)$.}\label{fig:shortest_path_tree_property}
\end{figure}

By Lemma~\ref{lem:ATC_even.compute_strong_visibility_each_edge}, we can compute $\interval{k/2}$ in $O(k^2\lambda_{k-1}(n)\log n)$ time.
Within the interval for each edge $vv'$ of $T$ such that $x(\g{k/2}) <  x(v) < x(v') < x(\g{k/2+1})$ and it is not fully visible from $\g{k/2}$ or $\g{k/2+1}$, we compute the minimum $y$-coordinate $h$ such that $e$ is covered by $\g{k/2}$ and $\g{k/2+1}$ together on $\altitude{h}$.
The optimal $y$-coordinate $h^*$ is then the maximum among all such values, which we can compute in $O(n)$ time using an argument similar to the proof of Theorem~\ref{thm:ATC_two}.
Therefore, we conclude this section with Theorem~\ref{thm:ATC_even}.

\begin{theorem}\label{thm:ATC_even}
For any even number $k > 2$, 
the Altitude Terrain Cover problem can be solved in $O(k^2\lambda_{k-1}(n)\log n)$ time.  
\end{theorem}

% \subsection{$k = 3$}
\subsection{Odd numbers of guards}
\label{sec:odd}
We consider odd numbers $k\geq 3$.
As in Section~\ref{sec:ATC_even},
we place the guards $\g{1},\ldots,\g{\lfloor k/2 \rfloor}$ at $\fFuncI{1},\ldots,\fFuncI{\lfloor k/2 \rfloor}$, 
and place $\g{k},\ldots,\g{k-\lfloor k/2 \rfloor+1}$ at $\gFuncI{1},\ldots,\gFuncI{\lfloor k/2 \rfloor}$.
Let $m = \lceil k/2 \rceil$ be the middle index.
By Lemma~\ref{lem:ATC_even.guard_location_necessary_recursive}, the optimal $y$-coordinate $h^*$
is the minimum value $h$ such that $x(\fFuncI{m}) = x(\gFuncI{m})$, and we can compute it efficiently using quasiconvex programming~\cite{amenta1999optimal}.

\begin{theorem}\label{thm:ATC_odd}
For any odd number $k\geq 3$, the Altitude Terrain Cover problem can be solved in $O(k^2\lambda_{k-2}(n)\log n)$ time.  
\end{theorem}
\begin{proof}
We first compute $\interval{\lfloor k/2 \rfloor}$ in $O(k^2\lambda_{k-2}(n)\log n)$ time by Lemma~\ref{lem:ATC_even.compute_strong_visibility_each_edge}.
Because $\fFuncI{m}$ (or $\gFuncI{m}$) can be represented as the upper envelope of monotone increasing (or decreasing) rational functions,
we compute the lowest vertex of the upper envelopes using quasiconvex programming~\cite{amenta1999optimal}. The $y$-coordinate of the vertex is the optimal value $h^*$. Such a lowest vertex can be found in $O(kn)$ time after an $O(k^2\lambda_{k-2}(n)\log n)$-time preprocessing.
\end{proof}

\section{Bijective altitude terrain cover}\label{sec:BATC_main}

In this section, we present algorithms for 
the bijective altitude terrain cover problem.
In Section~\ref{sec:BATC_(1)_main}, we present 
an $O(n)$-time
algorithm for case (1). 
In Section~\ref{sec:BATC_(2)_main}, we present an $O(kn)$-time algorithm for case (2). 

The following lemma shows that, for $k = 1$, 
there is an optimal placement of the guard $\g{}$
in the guard minimization version of the altitude terrain cover problem 
on the $x$-interval defined by the input terrain.
Let $\mathcal{E} \coloneq \bigcap_{e \subset T}e^+$ be the intersection of the (closed) upper half-planes defined by the edges of $T$.

\begin{lemma}\label{lem:BATC.lowest_guard_location_main}
For any $h\in\initInterval$, 
if $T$ is covered by a single guard $u$ on $\altitude{h}$, then there is a point $u'$ on $\altitude{h}$ with $x(v_1) \leq x(u') \leq x(v_n)$ that covers $T$.
\end{lemma}
\begin{proof}
A point covers $T$ if and only if it is contained in $\mathcal{E}$.
Suppose there is such a guard $u$ as in the lemma statement.
Then $\mathcal{E} \cap \altitude{h} \ne \varnothing$.
For the lowest vertex $w$ of $\mathcal{E}$, we have
$y(w) \leq y(u)$ and $x(v_1) \leq x(w) \leq x(v_n)$, because $\mathcal{E}$ 
is defined by the edges of $T$.
The horizontal projection $u'$ of $w$ on $\altitude{h}$, which is also contained in $\mathcal{E}$ because $\mathcal{E}$ is convex, covers $T$ by Lemma~\ref{lem:guard.monotone}.
\end{proof}

For any two points $p,q \in T$ with $x(p)\le x(q)$, 
we use $\subT{p}{q}$ to denote the subchain of $T$ from $p$ to $q$.
Let $(\g{1},T_1),\ldots,(\g{k},T_{k})$ be an optimal solution (\emph{guard-subchain pairs}) of the bijective altitude terrain cover problem.
By the above lemma,
each $\g{i}$ on $L(h)$ for $i = \{1, \dots, k\}$ can be placed such that $x(a_i) \leq \g{i} \leq x(b_i)$
for $T_i=T(a_i,b_i)$.

Moreover, there always is an optimal solution corresponding to an \emph{edge partition}: every endpoint of the $k$ subchains is a vertex of $T$.
This is verified in the following lemma.

\begin{lemma}\label{lem:BATC.edge_fully_covered_main}
For any optimal guard-subchain pair
$(\g{1},T_1),\ldots,(\g{k},T_{k})$, 
every endpoint of $T_i$, for each $i = 1, \dots, k$, is a vertex of $T$. 
\end{lemma}
\begin{proof}
Suppose not, and let $T_i$ for $1 < i \leq k$ be a subchain of $T$ 
with left endpoint $w$, which is not a vertex of $T$.
There is no loss of generality.
Let $v_jv_{j+1}$ be the edge of $T$ on which $w$ lies.
By Lemma~\ref{lem:BATC.lowest_guard_location_main}, we can move $\g{i}$ to a point $u'_i$ on the altitude line such that $u'_i$ lies on the $x$-interval defined by $T_i$, and $u'_i$ covers $T_i$. If $u'_i$ covers $v_jw$, then we are done by updating $T_{i-1}$ to $T_{i-1} \setminus v_iw$ and $T_i$ to $T_i \cup v_iw$.
Otherwise, there is $p \in v_{j}w$ not visible from $u'_i$, and there is a vertex $v$ of $T$ such that 
$x(v_{j+1}) < x(v) < x(u'_i)$ and $p$ lies below $\overline{vu'_i}$.
Then any point between 
$T(p, w) \subset T_i$
is not visible from $u'_i$, 
a contradiction.
\end{proof}

\subsection{Guard Optimization}\label{sec:BATC_(1)_main}

Similar to the algorithm by Daescu et al.~\cite{daescu2019altitude} for the altitude terrain guarding problem,
we greedily place guards on the given line $\altitude{h}$ from left to right.

Let $\IHP{i}{j} \coloneq \bigcap_{e \subset \subT{i}{j}}e^+$ be the intersection of the upper half-planes defined by the consecutive edges $v_iv_{i+1}, \ldots, v_{j-1}v_j$ of $T$.
For the leftmost guard, 
we find the largest index $j_1$ such that 
$\IHP{1}{j_1} \cap \altitude{h} \ne \varnothing$. 
By Lemma~\ref{lem:BATC.edge_fully_covered_main}, 
we can optimally place $\g{1}$ on any point of $\IHP{1}{j_1} \cap \altitude{h}$ to cover $\subT{1}{j_1}$.
Likewise, for every $i = 2, \ldots, k$, we 
find the largest index $j_i$ such that 
$\IHP{j_{i-1}}{j_i} \cap \altitude{h} \ne \varnothing$, and place $u_i$ on any point of it. We repeat this process until the whole terrain $T$ is covered by the guards.

By Lemmas~\ref{lem:guard.monotone} and~\ref{lem:BATC.lowest_guard_location_main}, 
instead of computing each $\IHP{j_{i-1}}{j_i}$, we compute its lowest vertex, which we denote by $w(j_{i-1},j_i)$.
For efficiency,  
our algorithm handles edges of positive slopes and edges of negative slopes separately while computing such lowest vertices. With edges of positive slope, it incrementally constructs the intersection $\mathcal{E}'$ of the corresponding half-planes. 
On the other hand, it handles edges of negative slopes one by one in order: compute the intersection of the edge and $\mathcal{E}'$, check if the lowest vertex changes, and update $\mathcal{E}'$ accordingly if so.

The following lemma shows 
how to efficiently compute such $w(j_{i-1},j_i)$'s as mentioned above.
Let $\IPHP{i}{j}$ and $\INHP{i}{j}$ be the intersection of $e^+$'s for all edges $e$ of $T$ with positive and negative slopes, respectively.

\begin{lemma}\label{lem:BATC.lowest_verex_all_indices_main}
For any indices $i, j$ with $i<j$, we can compute $w(i,\ell)$'s for all $(i+1) \leq \ell \leq j$ in $O(j-i)$ time.
\end{lemma}
\begin{proof}
Let $R(i,\ell) \coloneq \{p \mid p \in \IPHP{i}{\ell}, x(p)\geq x(w(i,\ell))\}$.
Let $e = v_{\ell}v_{\ell+1}$ be the edge at hand.
If $e$ has a positive slope, 
we have $w(i,\ell+1)=w(i,\ell)$ because $x(w(i,\ell)) \leq x(v_{\ell})$ by Lemma~\ref{lem:BATC.lowest_guard_location_main}.
Though, it may not be that $R(i,\ell+1) \ne R(i,\ell)$. This happens when $\overline{e}$ intersects $R(i,\ell)$. We can determine whether there is such an intersection 
by comparing the slopes of $e$ and the rightmost edge of $R(i,\ell-1)$.
If so, we can update $R(i,\ell)$ into $R(i, \ell+1)$ by trimming it with $\overline{e}$.

When $e$ has a negative slope, there are two cases. If $w(i,\ell)$ lies on or above $\overline{e}$,
we have $w(i,\ell+1) = w(i,\ell)$.
See Fig.~\ref{fig:contiguous_variant_algorithm_1}(a).
Otherwise, $w(i,\ell+1) = \overline{e} \cap R(i,\ell)$.
We find it by scanning the boundary of $R(i,\ell)$ from left to right.
After finding $w(i,\ell+1)$, we update $R(i, \ell)$ into $R(i, \ell+1)$ by trimming it 
with a vertical line $x = x\big(w(i,\ell)\big)$.
See Fig.~\ref{fig:contiguous_variant_algorithm_1}(b).

The time needed to handle all edges is linear in the number of edge insertions and deletions, while we maintain 
$R(i,\ell)$ for each $(i+1) \leq \ell \leq j$, which is $O(j-i)$.
\end{proof}

\begin{figure}[h]
  \centering
  \includegraphics[width=0.75\textwidth]{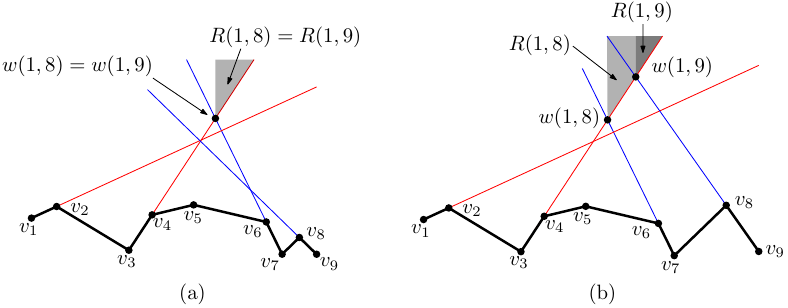}%
  \caption{
  (a) $w(1,8)=w(1,9)$ because $v_8v_9$ has a negative slope and $w(1,8)$ lies above $\overline{v_8v_9}$.
  (b) $w(1,9)=\partial R(1,8)\cap \overline{v_8v_9}$ because $v_8v_9$ has a negative slope and $w(1,8)$ lies below $\overline{v_8v_9}$, where $\partial R(1,8)$ is the boundary of $R(1,8)$.}\label{fig:contiguous_variant_algorithm_1}
\end{figure}

Recall that our algorithm finds, for each guard $u_i$, 
the longest subchain (corresponding to the index $j_i$) starting from the leftmost edge of the remaining parts of $T$.
By Lemma~\ref{lem:BATC.lowest_verex_all_indices_main}, we can find $j_1$ in $O(j_1)$ time, and each $j_i$ in $O(j_i-j_{i-1})$ time.
Thus, we can find an optimal set of guard-subchain pairs
in $O(n)$ time.

\begin{theorem}\label{thm:BATC_(1)_main}
Given a terrain with $n$ vertices and a horizontal line $L$, we can place the minimum number of guards on $L$ for case (1) of the Bijective Altitude Terrain Cover problem in $O(n)$ time.
\end{theorem}

\subsection{Altitude Optimization}\label{sec:BATC_(2)_main}
In this section, we consider case (2) of the bijective altitude terrain cover problem.

The case $k=1$ is exactly the same as in 
Section~\ref{sec:ATC}, thus it runs in $O(n)$ time.
For
$k = 2$,
we have $\min_{1 < i < n}\{\max\{y(w(1,i)), y(w(i,n)) \}\}$ 
as the optimal $y$-coordinate by Lemmas~\ref{lem:BATC.lowest_guard_location_main} and~\ref{lem:BATC.edge_fully_covered_main}.
By Lemma~\ref{lem:BATC.lowest_verex_all_indices_main}, we can compute
$w(1,i)$ and $w(i,n)$ for all $1 < i < n$ in $O(n)$ time.
Thus, we can compute the minimum $y$-coordinate in $O(n)$ time.

Suppose there are more than two guards.
By Lemmas~\ref{lem:BATC.lowest_guard_location_main} and~\ref{lem:BATC.edge_fully_covered_main}, 
we have the following.
\[h^* = \min_{1 < i_1 < \ldots i_{k-1} < n}\{\max\{y(w(1,i_1)),\ldots y(w(i_{k-1},n)) \}\}\]
A na\"ive approach would be to compute
$w(i,j)$ for all pairs $(i,j)$ with $1\leq i < j \leq n$. This takes $O(n^2)$ time, using Lemma~\ref{lem:BATC.lowest_verex_all_indices_main}.
Then $h^*$ can be found in $O(n)$ time with a greedy approach as in Section~\ref{sec:BATC_(1)_main}.
Finally, an optimal partition of $T$ can be found in $O(kn^2)$ time using dynamic programming.

We can improve the time complexity to $O(kn)$ by not computing all $w(i,j)$'s but only $O(n)$ candidates for each of the $k$ guards.
Our algorithm places guards from left to right in order.
It maintains $O(n)$ subintervals of $\initInterval$ such that the set of covered edges so far is the same for any value of $h$ in each subinterval.

\begin{comment}
 For the leftmost guard $u_1$, 
the algorithm computes $\mathcal{E}^+$
with the $O(n)$-time algorithm 
as in Section~\ref{sec:BATC_(1)}.
It then finds the points $\overline{e} \cap \partial\mathcal{E}^+$ for each edge $e$ of negative slope \ccheck{to compute lowest vertex $\mathcal{E}$}.
\end{comment}
For the leftmost guard $u_1$, 
the algorithm computes $w(1,2),\ldots,w(1,n)$ with 
the $O(n)$-time algorithm 
as in Section~\ref{sec:BATC_(1)_main}.
It is easy to see that within any of the 
$y$-intervals, each defined by one of 
$w(1,2),\ldots,w(1,n)$,
the set of covered edges is the same.
See Fig.~\ref{fig:contiguous_partition_interval}.

Then, for each remaining guard $u_i$,
the algorithm computes $\mathcal{E}^+(j_i, n)$ incrementally from the topmost interval to the bottommost interval, where 
$j_i$ is the smallest vertex index on $T_{i-1}$
(covered by $u_{i-1}$) on the bottommost interval. 
After that, it again finds the points $e \cap \mathcal{E}^+(j_i, n)$ for each edge $e$ of negative slope from the not yet covered portion of $T$.
These points further divide the 
$y$-intervals for $u_1$.
The number of intervals (for any $u_i$) is $m = O(n)$. We denote them by $I_{i, 1}, \ldots, I_{i, m}$ from bottom to top.
The algorithm stops when it finds a subinterval such that the rightmost guard covers the rightmost edge of $T$.
It takes $O(n)$ time for each guard, since we can compute $\mathcal{E}^+(j_i, n)$ 
in $O(n)$ time, and we can find all intersection points $e \cap \mathcal{E}^+(j_i, n)$ in the same time bound.

\begin{figure}[h]
  \centering
  \includegraphics[width=0.5\textwidth]{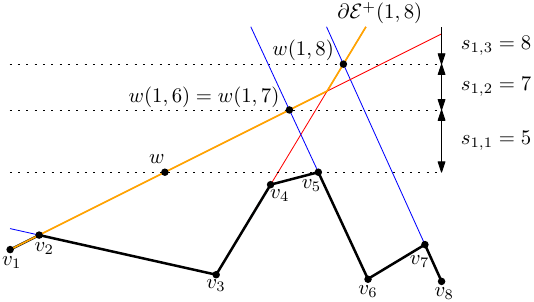}%
  \caption{
  The orange chain $\partial\IPHP{1}{8}$ is the boundary of $\IPHP{1}{8}$.
  The rightmost point of $\IPHP{1}{8}\cap\altitude{\ymax}$ is $w$.
  We have $s_{1,1}=5$ because $v_5v_6$ is the leftmost edge of $T$ such that $w$ lies below $\overline{v_5v_6}$. $s_{1,2}=7$ because $w(1,6)=w(1,7)$ and $w(1,7)\neq w(1,8)$.}\label{fig:contiguous_partition_interval}
\end{figure}

Let $I_{1, 1}, \ldots, I_{1, m}$ be the 
$m$ subintervals for $u_1$. 
For each $j = 1, \ldots, m$, 
we denote by $s_{1, j}$ the largest index such that 
$\subT{1}{s_{1, j}}$ is covered by $\g{1}$ in $I_{1, j}$. 
Similarly, for any $1 < i \leq k$, let $I_{i, 1}, \ldots, I_{i, m'}$ be the 
$m' \geq m$
subintervals for $u_i$. 
For each $j = 1, \ldots, m'$,
we denote by $s_{i, j}$ the largest index such that
$\subT{s_{i-1, j}}{s_{i, j}}$ is covered by $\g{i}$ in $I_{i, j}$.
Such $s_{i, j}$ is well-defined because the set of subintervals for $u_{i}$ is a refinement of the set of subintervals for $u_{i-1}$ for any $1 < i \leq k$.
Then the optimal $y$-coordinate $h^*$ is the 
lower limit point of the $y$-interval in which $v_n$ is covered by $u_k$.

The following lemma verifies that 
if we 
compute 
the intersection of the upper half-planes (for some positive edges of $T$) for each the bottommost interval only once,
we can use them to compute the lowest vertices for multiple values of $j$.

\begin{lemma}\label{lem:BATC.lowest_verex_properties}
For any indices $i,j$ with $i < j \leq n$, $w(i,j)$ is equal to the lowest vertex of $\IPHP{i}{n} \cap \INHP{i}{j}$. 
\end{lemma}
\begin{proof}
For any edges $e, e'$ of $T$ such that 
$e\in \INHP{i}{j}$ and $e'\in \IPHP{j}{n}$,
$\overline{e} \cap \overline{e'} \notin \tilde{T}$. 
\end{proof}

The following lemma shows that our algorithm uses $O(n)$ time for each guard.

\begin{lemma}\label{lem:BATC.computing_indices_recursively}
For any integer $i$ with $1<i\leq k$,
given subintervals
$I_{i-1, 1}, \ldots, I_{i-1, m}$ 
and the corresponding indices 
$s_{i-1, 1}, \ldots, s_{i-1, m}$,
we can compute
$I_{i, 1}, \ldots, I_{i, m'}$, and 
$s_{i, 1} \ldots, s_{i, m'}$ in $O(n)$ time.
\end{lemma}
\begin{proof}
Let $P_{i, j}$ be the restriction of $\mathcal{E}^+(s_{i-1, j}, n)$ on $I_{i-1,j}$. 
We compute incrementally $P_{i, j}$ from $j=m \ldots, 1$ in order as follows.
For $j = m$, we compute $\mathcal{E}^+(s_{i-1,m},n)$ by the edges of $T$ with positive slopes are considered from right to left.
We partition $\mathcal{E}^+(s_{i-1,m},n)$ into $P_{i, m}$ and  $\mathcal{E}^+(s_{i-1,m},n)\setminus P_{i, m}$ by trimming it with the horizontal line through the lower boundary of $I_{i-1, m}$.
We compute $P_{i, m-1},\ldots,P_{i,1}$ 
in a similar way. This takes linear time. See Figure~\ref{fig:incremental_algorithm_envelope}.

For each $1 \leq j \leq m$, we compute $w(s_{i-1, j},\ell)$ for all $s_{i-1,j} < \ell \leq n$ such that $w(s_{i-1, j},\ell)$ is contained in $I_{i-1,j}$.
By Lemma~\ref{lem:BATC.lowest_verex_properties}, we compute the intersection between $P_{i, j}$ and the lines extended from edges of $T(s_{i-1, j}, n)$ of negative slopes instead of computing $\mathcal{E}(s_{i-1, j},\ell)$. 
From $j=1, \ldots, j=m$ in order, we compute the intersection of $P_{i, j}$ and the lines extended from edges of $T(s_{i-1, j}, n)$ with negative slopes, using an argument similar to the proof of Lemma~\ref{lem:BATC.lowest_verex_all_indices_main}.
We consider such edges from left to right.
If we encounter the first edge $e$ with negative slope such that 
the intersection of $e$ and $P_{i, j}$ lies above $I_{i-1,j}$, then we update $j$ into $j+1$, initialize the intersection, and recompute it starting from $e$.
In this way, we can compute $I_{i, 1}, \ldots, I_{i, m'}$, and 
$s_{i, 1} \ldots, s_{i, m'}$ in $O(n)$ time.
\end{proof}

\begin{figure}[h]
  \centering
  \includegraphics[width=0.9\textwidth]{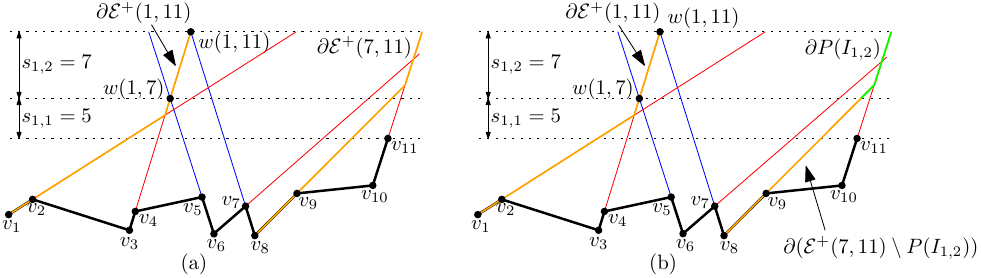}%
  \caption{(a) $\partial\mathcal{E}^+(1,11)$ is the boundary of $\mathcal{E}^+(1,11)$. $\partial\mathcal{E}^+(7,11)$ is the boundary of $\mathcal{E}^+(7,11)$. $\IPHP{s_{1,2}}{n}=\IPHP{7}{11}$. (b) Computation of $\partial \IPHPInterval{I_{1, 2}}$ using $\IPHP{7}{11}$, where $\partial \IPHPInterval{I_{1, 2}}$ is the boundary of $\IPHPInterval{I_{1, 2}}$.}
  \label{fig:incremental_algorithm_envelope}
\end{figure}

By Lemma~\ref{lem:BATC.computing_indices_recursively},
we can compute all $s_{i, j}$'s for guards $u_i$ and the corresponding subinterval $I_{i, j}$, and find the optimal $y$-coordinate in $O(kn)$ time. Thus, we have Theorem~\ref{thm:BATC_(2)_3_main}.

\begin{theorem}\label{thm:BATC_(2)_3_main}
Given a terrain with $n$ vertices and a positive integer $k$, we can place the lowest horizontal line $L$ and place $k$ guards on $L$ for case (2) of the Bijective Altitude Terrain Cover problem 
in $O(kn)$ time.
\end{theorem}

% \subsection{Altitude Optimization}\label{sec:BATC_(2)_main}

% ---- Bibliography ----
\small
\bibliographystyle{alphaurl} % arxiv
\bibliography{bib}
%\begin{thebibliography}{10}
%\end{thebibliography}

\end{document}